\documentclass[12pt,english,leqno]{article}
\usepackage{mathptmx}
\usepackage[T1]{fontenc}
\usepackage[unicode=true,pdfusetitle,
bookmarks=true,bookmarksnumbered=false,bookmarksopen=false,
breaklinks=false,pdfborder={0 0 0},pdfborderstyle={},backref=false,colorlinks=true] {hyperref}
\hypersetup{
	citecolor=blue, linkcolor=blue, urlcolor=blue}
\usepackage{geometry}
\usepackage[utf8]{inputenc}
\usepackage{amsmath,mathtools,commath}
\usepackage{amsthm}
\usepackage{amssymb}
\usepackage[authoryear,round]{natbib}
\usepackage{hyperref}
\usepackage{amsfonts}
\usepackage{color}
\usepackage{array}
\usepackage{verbatim}
\usepackage{multirow}
\usepackage{graphicx}
\usepackage{setspace}
\usepackage{hyperref}
\usepackage{epigraph}
\usepackage{subfigure}
\usepackage{accents}
\usepackage{tikz}
\usepackage{tikz-cd}
\usepackage{tkz-euclide}
\usepackage[makeroom]{cancel}
\usepackage{geometry}
\usepackage{mathrsfs}
\usepackage{bm}
\usepackage{natbib}
\usepackage{appendix}

\def\b{\beta}

\def\:={\coloneqq}

\newlength{\dhatheight}

\def\b{\beta}

\def\:={\coloneqq}

\newtheorem{theorem}{Theorem}[section]
\newtheorem{proposition}{Proposition}[section]
\newtheorem{lemma}[proposition]{Lemma}
\newtheorem{corollary}[proposition]{Corollary}
\newtheorem{definition}[proposition]{Definition}
\theoremstyle{remark}
\newtheorem{remark}[proposition]{Remark}

\title{Automation, Savings, and Breakaway to a Post-Labour Economy
\thanks{We thank the organisers and participants of the 2026 Silvaplana
  Political Economy Workshop for helpful comments and discussions.
  Hellman gratefully acknowledges financial support from the Israel
  Science Foundation grant 458/25.
  The authors declare no competing interests.}}
\author{Omer Edhan and Ziv Hellman}
\date{This Draft: \today}

\begin{document}
	
	\maketitle

\begin{abstract}
When does automation become an economically self-sustaining process? We define
\emph{automation breakaway} as the condition under which saving generated by
physical capital and artificial labour is sufficient to finance continued
automation. Artificial labour is labour-like in production but capital-like in
accumulation, creating an income base that can support further investment. We
characterise breakaway and its dynamic viability, showing that productivity and
factor-income composition provide distinct financing channels. Artificial-labour
accumulation strengthens the saving base supporting breakaway, while its effect on
labour income changes sign with substitutability with human labour. Near full
automation, aggregate labour income can contract through fewer human-filled jobs
while average labour income in surviving jobs remains positive or rises. Optimal
reliance on worker finance can begin to decline while it still provides a positive
current net benefit.
\medskip

\noindent\textbf{JEL Classification:} C61, E21, E22, E25, J23, O33
\end{abstract}

\section{Introduction}
\label{sec:introduction}

When does automation become an economically self-sustaining process? Advances in
artificial intelligence and robotics increasingly allow production to be carried
out by accumulable technologies in place of human labour. 
It is reasonable to assume that this trend will continue in the future, perhaps even at an accelerated pace, defining a transitional period between economies dependence on human labour and economies in which production is primarily accomplished by autonomous machines.

Such a fundamental economic transition
will require continuing investment to proceed. During the transition period, however, part of that
investment may still need to be financed by saving that is generated from labour income, even as labour income may be reduced by artificial labour replacing human workers. Automation
therefore changes the financing base of further automation at the same time it
changes the productive role of labour. 
The upshot is that a rapid automation process can cause its own unsustainability by reducing human labour income (and hence savings) to such an extent that it undermines its financing base. 
The resulting accumulation problem is
simple: when does the automation process generate enough resources to sustain its own continued development, without reliance on saving from human labour income?

This paper studies that financing transition. Human labour plays two economic
roles along the path of automation. 
It contributes productive services, and also savings that accumulate and contribute to investment in the economy.
A process of automation may diminish labour's productive role faster than its needed role in financing accumulation.
An economy can therefore become increasingly automated while the investment programme
supporting that automation continues to draw on worker saving. Such a transition, however,
reaches a qualitatively different regime once physical capital and artificial
labour generate enough saving on their own to maintain the automation system and finance its
continued advance from their own income. At that point, automation can break free of any dependence on saving from human labour. We call this
\emph{automation breakaway}.

Breakaway turns the transition toward a post-labour economy into an accumulation
problem. Prior to breakaway, worker saving contributes to the financing base of
continued automation. After breakaway, the capital--artificial-labour accumulation
circuit has sufficient financing capacity to sustain the programme from the income
generated by accumulable factors. Automation can then finance the continued
expansion of the productive system that increases its own role in production. The
central question is when accumulation acquires this financial autonomy and whether
that autonomy persists as automation advances.

Artificial labour gives this problem its distinctive economic structure.
Artificial intelligence can increasingly perform work-like functions while
remaining reproducible, ownable, and expandable through investment. It
therefore incorporates elements of both labour and capital together: labour-like in production, and capital-like in accumulation. Recent work
on autonomous AI similarly emphasises the scalability of artificial agents across
available computational resources, in contrast with human knowledge and time,
which remain embodied in individual workers \citep{ide2025artificial}. This dual
character allows artificial labour to substitute for human labour inside
production while joining physical capital as part of the accumulable income base
that finances further investment.

We develop a tractable dynamic framework built around this distinction. Production
uses physical capital, human labour, and artificial labour. 
This receives explicit expression in the formulae of the model. Where in  standard models output $Y$ is a function of capital $K$ and human labour $L$, we add here a third major variable, $\Lambda$, representing effective artificial labour capacity and its aspects separate from traditional capital and human labour. Improvements in the quantity or quality of artificial labour have direct effects on output in our model, and are reflected in $\Lambda$.

An automation frontier
determines the share of the task domain currently performed by artificial labour, while AI
capability can also augment human productivity in tasks that remain
human-essential. Capital owners, workers, and owners of artificial labour save out
of their respective income streams. Continued automation requires investment to
maintain productive stocks, expand artificial-labour capacity, improve AI
capability, and advance the automation frontier. This structure allows the
financing capacity of automation and the economic position of human labour to
evolve jointly while remaining analytically distinct.

Our first results characterise the emergence and persistence of financial autonomy.
Breakaway is reached when saving generated by capital and artificial labour covers
the investment required to continue the automation programme. Dynamic viability
requires this financing capacity to reproduce itself as automation advances. In
the benchmark dynamics, an advance of the automation frontier must generate
sufficient additional non-labour saving to cover the accompanying increase in the
investment requirement. Breakaway therefore describes a self-sustaining
saving-investment regime, rather than a particular technological level.

The production structure reveals two distinct forces behind this financing
condition. Frontier automation can raise total output, expanding the income
available for saving. It can also change the composition of factor income, shifting
income towards artificial labour and therefore toward the income base that finances
breakaway. The output effect and the income-composition effect enter the viability
condition separately. This yields one of the paper's main results: frontier
automation can strengthen viable breakaway even when its marginal effect on total
output is negative. A sufficiently strong shift in factor income toward the
non-labour saving base can finance the additional investment requirement despite a
contemporaneous decline in output. Financial autonomy and marginal productivity
therefore respond to distinct economic margins.

Artificial-labour accumulation provides a second central result. Expanding the
artificial labour stock always expands the non-labour saving base. Its effect on
labour income, in contrast, depends on the productive relationship between human
and artificial labour. With sufficiently strong substitutability,
artificial labour accumulation reduces labour income; with sufficiently weak
substitutability, it raises labour income. The same accumulation process that
strengthens the financing base of automation can therefore support very different
paths for workers. This separates the transition to breakaway from the evolution
of labour income and shows how both margins are shaped by the productive role of
artificial labour.

We then characterise the long-run evolution of labour income. Displacement reduces
the domain of tasks for which human labour remains essential, augmentation raises
human productivity within that domain, and accumulation changes productive scale
and the contribution of the other inputs to human labour income. The cumulative balance of all of these effects determines
whether labour income contracts along a path approaching full automation, or
remains economically significant. Paths approaching full automation can therefore
retain economically significant labour income when augmentation and favourable scale
effects are sufficiently strong. Breakaway and labour-income persistence describe separate
dimensions of the transition: one concerns the financing autonomy of accumulation,
while the other concerns the income generated by human productive activity.

The task structure provides a further mapping from aggregate labour income to
human work. Jobs consist of bundles of core tasks, and a job remains human-filled
whenever its bundle retains a human-essential requirement. In the finite-mean
benchmark, frontier expansion reduces both aggregate labour income and the mass of
human-filled jobs through the same direct extensive-margin force. That common
frontier effect therefore cancels from average labour income per remaining
human-filled job. Augmentation and the surrounding productive environment then
determine how income in the jobs that remain evolves. Aggregate labour income can
consequently contract through a shrinking mass of human-filled jobs while income
associated with the surviving jobs remains stable or rises.

The final part of the paper makes the financing margin itself an optimal dynamic
choice. Worker saving is a source of finance external to the
capital--artificial-labour accumulation circuit, and the investor determines the
extent and duration of reliance on it. Worker finance supplies resources for
investment, reliance on that source carries a reduced-form external-finance wedge,
and changing the financing structure carries an adjustment cost. The financing
structure of automation therefore becomes part of optimal investment behaviour.

The optimal financing decision is forward-looking. The value of continued reliance
on worker finance depends on the future path of labour income and on the future
cost of external finance. This creates an anticipatory reduction in reliance: an
optimal programme can begin reducing its reliance on worker finance while worker
saving still provides a positive current net financing benefit. Under stronger
conditions, reliance reaches zero in finite time. The financing margin therefore remains
economically substantive inside an optimal investment problem: the programme may use
worker finance during the transition while building financial independence in
anticipation of the erosion of the income base supporting that finance.

These results add an accumulation margin to the production-based analysis of
automation. Task-based models explain how technology changes the productive role
of labour, while transformative-growth models characterise economies in which
increasingly capable AI relaxes human labour as a productive constraint. The
breakaway perspective studies the resources required to sustain movement along
that transition. It connects the task-based logic of automation with the classical
saving-investment logic of accumulation in an economy in which artificial labour
is itself an accumulable productive factor. In this sense, our analysis also
complements \citet{restrepo2025missed}, who characterises the economic structure of an AGI
limit in which economically valuable work can increasingly be performed with
compute. We characterise the financing conditions governing movement toward such
a regime.

The contribution of the paper is a theory of automation breakaway. It identifies
when automation becomes a financially self-sustaining accumulation process,
characterises the dynamic condition that preserves this regime, and shows how
output and factor-income composition jointly determine its viability. The framework
then separates the financing transition from the evolution of labour income and
human work, and shows how reliance on worker finance enters an optimal dynamic
investment problem. It thereby distinguishes three dimensions of the transition
toward a post-labour economy: the technological capacity of automation, the
financial autonomy of the accumulation process, and the continuing economic role
of human labour.

The paper proceeds as follows. Section~\ref{sec:framework} introduces the
production environment, factor incomes, saving flows, and the investment
requirements of continued automation. Section~\ref{sec:breakaway-viability}
introduces automation breakaway and characterises its dynamic viability.
Section~\ref{sec:labour-income-jobs} studies labour income, artificial-labour
accumulation, and the extensive margin of human-filled jobs.
Section~\ref{sec:optimal} studies optimal reliance on worker finance.
Section~\ref{sec:conclusion} concludes.

\section{Related Literature}
\label{sec:related_literature}

Our analysis relates most directly to three literature strains: task-based models of
automation and artificial labour, work on AI-driven growth and increasingly
machine-based production, and theories of saving, distribution, and accumulation.
These literatures provide the production and accumulation ingredients of our
framework. We organise them around a distinct dynamic object: the financing
capacity of continued automation. The central question is when the income generated
by accumulable factors provides sufficient saving to sustain the investment
required for further automation.

The task-based automation literature provides the production-side representation
of technological displacement. Early work by \citet{zeira1998workers} studies the
replacement of labour requirements by machine capital, while task-based models
make the allocation of activities between labour and machines explicit
\citep{acemoglu2011skills,acemoglu2019automation,acemoglu2018ai}. We use this
structure to represent the contraction of the human-essential task domain and carry
that reallocation into the accumulation problem. Labour income generates worker
saving, so a change in the productive role of labour also changes one source of
finance for continued automation. This gives task displacement a financing margin
alongside its effects on production, wages, and factor shares. The model also
allows AI capability to augment productivity within tasks that remain
human-essential, consistent with the broader literature on skill-complementary
technical change \citep{acemoglu2002directed,acemoglu2002technical} and recent
evidence that generative AI can raise worker productivity
\citep{brynjolfsson2025generative,noy2023experimental}. Related task-exposure
exercises quantify the range of work activities to which large language models may
apply \citep{eloundou2024gpts}.

Recent work by \citet{ide2025artificial} provides a particularly useful
production-side interpretation of artificial labour. They model autonomous AI as
scalable agents capable of performing knowledge work that could otherwise be
performed by humans and study how AI capability and available compute reshape the
organisation and allocation of knowledge work. Our framework places the
accumulation of this productive capacity at the centre of the analysis:
artificial-labour capacity is an expandable stock whose growth requires investment
and whose use generates factor income and saving. This turns the scalability of
artificial labour from a production property into an accumulation mechanism. The
reproducibility margin is also related to \citet{raymookherjee2022growth}, who show
how the asymmetry between replicable physical capital and embodied human capital
can generate progressive automation through accumulation. In our setting,
artificial labour combines labour-like productive functions with capital-like
accumulation and ownership. Its expansion therefore enlarges the non-labour saving
base, while its effect on labour income depends on the productive relationship
between human and artificial labour.

A second literature studies artificial intelligence and growth as increasingly
capable machine inputs relax human labour as a productive constraint
\citep{aghion2019ai,nordhaus2021singularity}. A particularly close growth-side
comparison is \citet{trammell2025growth}, who characterise the consequences of
transformative AI for growth and show how sufficiently extensive production
automation can make machine inputs economically self-replicating.
\citet{restrepo2025missed} characterises the productive allocation as expanding
compute makes an increasing range of economically valuable work feasible with
machines and the labour share approaches zero. Our analysis adds the financing
dimension of this transition. Breakaway identifies when the
capital--artificial-labour circuit generates enough saving to meet the investment
requirements associated with maintaining and advancing automation. Technological
self-replication and financial self-sustenance are thus distinct and complementary
features of an increasingly machine-based economy.

Related models of robots and machine capital show how automation interacts with
accumulation through wages, saving, wealth, and factor-income distribution.
\citet{sachs2012smart}, in particular, show how displacement can weaken subsequent
capital formation by reducing the income and saving capacity of workers.
\citet{benzell2015robots} and \citet{sachs2015robots} study the broader
macroeconomic and distributional effects of smart machines, while
\citet{berg2018fear} combines robot capital with heterogeneous saving behaviour and
\citet{moll2022uneven} links automation to the joint evolution of income and wealth
inequality. These papers establish several channels through which automation can
reshape accumulation. Our financing-regime margin follows the complementary
channel generated by an expanding non-labour income base: as capital and artificial
labour accumulate, the saving generated by those factors can acquire the capacity
to finance the continuation of automation from within the accumulable part of the
economy.

The saving mechanism builds on the classical literature on distribution and
accumulation. \citet{kaldor1955alternative} links aggregate saving to the
composition of wage and profit income through different saving propensities.
\citet{pasinetti1962rate} extends this framework by allowing workers who save to
accumulate capital and receive profit income, thereby making the interaction
between saving, ownership, and distribution explicit. Our framework uses
source-specific saving propensities to isolate how automation-induced changes in
factor-income composition affect financing capacity. Introducing artificial labour
adds an accumulable source of income whose importance grows with automation.
Changes in the automation frontier can then affect the financing capacity of the
economy through both total output and the composition of factor income. This
separation generates one of our central results: a shift of income toward the
non-labour saving base can strengthen the viability of breakaway even when frontier
advancement lowers output at the margin.

Recent work brings saving, ownership, and financial structure directly into models
of AI-driven accumulation. \citet{lagarda2026owns} study how ownership of AI
capital shapes household income, saving, and capital formation, while
\citet{bond2023income} study how heterogeneous financial frictions shape
distribution as automation becomes asymptotically complete. A particularly close
comparison is \citet{caballero2026speculative}. Caballero studies an equilibrium
in which AI-induced shifts toward capital income interact with wealth-dependent
saving and asset valuations and can support a self-sustaining high-capital regime.
Our breakaway criterion takes the current state and continuation programme as its
objects and asks whether saving generated within the
capital--artificial-labour circuit covers the associated continuation requirement.
Holding saving propensities by income source fixed makes the income-composition
channel transparent and permits the output and composition effects of automation
on financing capacity to be characterised separately. Dynamic viability then asks
whether this financing condition reproduces itself as the automation frontier
continues to advance.

Taken together, these literatures locate the contribution of the paper at the
financing transition created when a labour-performing technology is also an
accumulable and income-generating factor of production. Task-based models determine
how automation reallocates productive activity, AI growth models characterise the
expansion of machine-based production, and accumulation theory explains how the
resulting income streams support investment. Breakaway makes the financing capacity
of this process an explicit state-dependent object. The final part of the paper
extends the same financing margin to optimal investment behaviour by making
reliance on worker finance an intertemporal choice as the labour-income base evolves.

\section{Framework}
\label{sec:framework}

This section sets out the economic environment and the investment problem studied
in the paper.  The model has three productive inputs: physical capital, human
labour, and artificial labour.  Artificial labour is the economic form introduced by
AI.  It is labour-like in production, because it performs tasks that would otherwise
require human workers, but it is capital-like in accumulation and ownership,
because it is reproducible, accumulable, and ownable.

\subsection{Production with artificial labour}
\label{subsec:production-artificial-labour}

The economy combines three productive inputs: physical capital \(K\), human labour
\(L\), and artificial labour \(\Lambda\).  Output is produced according to
\begin{equation}
	\label{eq:general-production-three-factors}
	Y=F(K,L,\Lambda;\beta,a),
\end{equation}
where \(\beta\in[0,1]\) is the automation frontier and \(a\geq0\) is augmentation
capability.  The automation frontier measures the share of the task domain currently
automated and performed by artificial labour.  Augmentation capability raises the productivity of human labour
in tasks where human input remains essential.  We interpret \(\Lambda\) as effective
artificial-labour capacity, so improvements in the quantity or quality of artificial
labour are embodied in \(\Lambda\).

The state of the automation system is
\begin{equation}
	\label{eq:state-model}
	x=(K,\Lambda,a,\beta).
\end{equation}
These are the accumulable and technologically advancing components whose
maintenance and expansion define the investment problem.

The production structure combines two margins.  Human labour and artificial labour
are averaged into a labour mean, with weights determined by the automation
frontier.  The weight \(1-\beta\) is the human-essential task domain, while
\(\beta\) is the automated task domain.  Capital then combines with this labour
mean at the aggregate production level.  This means that artificial labour substitutes for
human labour inside the labour mean, while physical capital combines with the
labour mean as a whole.

Define the \textit{labour mean} by
\begin{equation}
	\label{eq:labour-mean}
	M \:=
	\left[
		(1-\beta)a^\theta L^\rho+\beta\Lambda^\rho
	\right]^{1/\rho}.
\end{equation}
The object \(M\) is a mean of human and artificial labour services.  The parameter
\(\rho\in(0,1)\) governs substitution between these two sources of labour services
inside \(M\).  This is distinct from the task margin governed by \(\beta\): as
\(\beta\) rises, greater weight is placed on artificial labour and less on the
human-essential task domain.  The term $a^\theta$, is a measure of 
augmentation capability: the higher $a^\theta$, the greater the effective contribution of human labour in the
tasks where human input remains essential.

Throughout this paper, we use the following benchmark production function
\begin{equation}
	\label{eq:production-labour-mean}
	Y = K^\alpha M^{1-\alpha},
\end{equation}
where \(\alpha\in(0,1)\).  Written out more fully, this is
\begin{equation}
	\label{eq:ces-production-example}
	Y = K^\alpha
	\left[
		(1-\beta)a^\theta L^\rho+\beta\Lambda^\rho
	\right]^{\frac{1-\alpha}{\rho}}.
\end{equation}

We assume competitive factor pricing throughout. Let
\(w(x)=Y_L(x)\)
denote the competitive wage per unit of human labour. Factor incomes are therefore
\[
	Y^L=wL=L Y_L,\qquad
	Y^K=K Y_K,\qquad
	Y^\Lambda=\Lambda Y_\Lambda .
\]
For the benchmark above, factor incomes are
\begin{align}
	Y^K
	&=
	\alpha Y,
	\label{eq:capital-income-production-example}
	\\
	Y^L
	&=
	(1-\alpha)Y
	\frac{(1-\beta)a^\theta L^\rho}{M^\rho},
	\label{eq:labour-income-production-example}
	\\
	Y^\Lambda
	&=
	(1-\alpha)Y
	\frac{\beta\Lambda^\rho}{M^\rho}.
	\label{eq:artificial-labour-income-production-example}
\end{align}

These expressions connect production to financing.  Worker saving is generated by
labour income \(Y^L\).  Non-labour saving is generated by capital income \(Y^K\)
and artificial-labour income \(Y^\Lambda\).  As automation advances, the labour
mean places greater and greater weight on artificial labour; simultaneously, saving flows shift toward
the accumulable automation system.  Breakaway asks when those non-labour saving
flows generate non-negative continuation surplus when worker saving is excluded
from the investment circuit.

\subsection{Saving and the Automation Investment Problem}
\label{subsec:saving-investment-problem}

Competitive factor incomes determine the saving flows available for investment.
The saving rate will be assumed to satisfy $(\sigma^j \in [0,1]$, for each income source \(j \in \{L,K,\Lambda\}\).

\emph{Aggregate saving} is
\begin{equation}
	\label{eq:aggregate-saving-model}
	\Sigma(x)
	=
	\sigma^L Y^L(x)+\sigma^K Y^K(x)+\sigma^\Lambda Y^\Lambda(x).
\end{equation}
\textit{Non-labour saving}, generated by capital and artificial labour, is
\begin{equation}
	\label{eq:non-labour-saving-model}
	\Sigma^E(x)
	=
	\sigma^K Y^K(x)+\sigma^\Lambda Y^\Lambda(x).
\end{equation}
\(\Sigma^E\) is the saving base generated by the accumulable factors $K$ and \(\Lambda\).

Sustained automation requires investment in all of the components of the automation
system: physical capital \(K\), artificial-labour capacity \(\Lambda\),
augmentation capability \(a\), and the frontier \(\beta\).  These investments
maintain productive capacity, expand artificial labour, expand augmentation
capability, and expand the task domain performed by artificial labour.

The common accounting environment supports two distinct questions regarding the role
of labour-income saving in accumulation.  The first is a financing-feasibility
question: does non-labour saving generate non-negative continuation surplus when
worker saving is excluded?  The second is an optimal-reliance question: when
current flow financing feasibility is deliberately relaxed, what is the optimal
path of automation investment and reliance on worker finance?  This distinction
separates whether worker finance is required for continuation from whether worker
finance remains optimally valuable when its use is not required for admissibility.

We first formalise the financing-feasibility question.  The optimal-reliance
question will be studied later, in Section \ref{sec:optimal}.  Let
\(I = (I^K,I^\Lambda,I^a,I^\beta)\) denote investment flows, and let the automation
state evolve according to the differential equation
\begin{equation}
	\label{eq:general-dynamics-model}
	\dot{x} = G(x,I).
\end{equation}
The continuation requirement for the frontier is expressed by a state-dependent
rate \(\eta^\beta(x)>0\).  Conditional on continuation of the automation programme,
\(\eta^\beta(x)\) is the minimum frontier advance that defines continuation.  We
regard this requirement as a reduced-form primitive.  One possible interpretation is
that it reflects investment opportunities outside the automation programme, mapped
into the frontier advance required for continuation.

For a given target frontier advance \(\eta^\beta(x)>0\), define the full investment
requirement by
\begin{equation}
	\label{eq:tfull-general-model}
	\mathcal{T}^{\mathrm{full}}(x)
	=
	\inf_{I\geq0}
	\left\{
		C(I;x) \mid
		G^K(x,I)\geq0,\;
		G^\Lambda(x,I)\geq0,\;
		G^a(x,I)\geq0,\;
		G^\beta(x,I)\geq\eta^\beta(x)
	\right\}.
\end{equation}
$\mathcal{T}^{\mathrm{full}}(x)$ is the minimum investment cost required
to maintain the automation system and advance the automation frontier.

The automation frontier \(\beta\) and the breakaway surplus express
the two forms of labour dependence introduced above.  The frontier concerns
production, while the breakaway surplus concerns accumulation.  An economy can therefore
approach a technologically post-labour production structure while still requiring
labour-income saving to finance continued automation.  Conversely, breakaway need
not require a nearly complete automation frontier: \(B(x) \geq 0\) may hold while at the same time
$\beta < 1$, hence accumulation may become independent of worker finance while
production still retains a human-essential task domain.

We define the \emph{breakaway surplus} to be
\begin{equation}
	\label{eq:breakaway-surplus-model}
	B(x)
	=
	\Sigma^E(x)-\mathcal{T}^{\mathrm{full}}(x).
\end{equation}
The object \(B\) is an investment surplus.  It measures the net continuation
surplus from carrying out the required automation investment using non-labour
saving alone.

\textit{Breakaway} occurs when
\begin{equation}
	\label{eq:breakaway-condition-model}
	B(x)\geq0.
\end{equation}
At such a state, the capital--artificial-labour subsystem generates enough saving
to maintain and expand automation without a contribution from worker saving.

The next section studies whether this condition is dynamically persistent:
breakaway and viability ask whether the non-labour saving-investment circuit can remain
closed over time.  Section~\ref{sec:labour-income-jobs} then studies what happens
to labour income and human-filled jobs as the economy approaches full automation.
Section~\ref{sec:optimal} uses that labour-income margin in a deliberately
permissive financing benchmark to study the distinct optimality question: whether
and how quickly reliance on worker finance should be reduced.

\section{Breakaway and Viability}
\label{sec:breakaway-viability}

This section studies the investment surplus question introduced in the framework.
The breakaway surplus \(B(x)\) measures non-labour saving net of the full investment
requirement of automation when worker saving is removed from the investment
circuit.  

\subsection{Breakaway}
\label{subsec:breakaway}

For \(m\geq0\), define
\begin{equation}
	\label{eq:Bm}
	B_m=\{x:B(x)\geq m\}.
\end{equation}
The set \(B_0\) is the breakaway region.  More generally, \(B_m\) is the region in
which the capital--artificial-labour subsystem generates a surplus of at least
\(m\) after covering the full investment requirement of automation.


Note that breakaway is a pointwise investment surplus condition.  Viability, in contrast, asks whether this
condition is preserved as the economy evolves.  To study persistence, fix an induced law of
motion
\(\dot{x}=g(x)\).
The vector field \(g\) may arise from a decentralised investment process, a policy
rule, or an investment rule \(I(x)\), via $g(x) = G(x,I(x))$.

\begin{definition}
	\label{def:viability}
	The set \(B_m\) is \emph{viable} under the dynamics \(\dot{x} = g(x)\) if every trajectory
	with \(x(0) \in B_m\) satisfies \(x(t) \in B_m\) for all \(t \geq 0\).
\end{definition}

\begin{theorem}
	\label{thm:viable-breakaway}
	Assume that \(g\) is Lipschitz, that \(B\) is continuously differentiable, and
	that \(m\) is a regular value of \(B\).  Then \(B_m\) is viable if and only if
	\begin{equation}
		\label{eq:nagumo-condition}
		\nabla B(x)\cdot g(x)\geq0
	\end{equation}
	for every state \(x\) satisfying \(B(x)=m\).
\end{theorem}

\begin{proof}
	Since \(B_m=\{x:B(x)\geq m\}\), its boundary is
	\(\partial B_m=\{x:B(x)=m\}\).  Since \(m\) is a regular value of \(B\), the
	boundary is a smooth hypersurface.  At each boundary point, the tangent cone is
	\(T_{B_m}(x)=\{v:\nabla B(x)\cdot v\geq0\}\).
	By Nagumo's theorem \citep{nagumo1942,aubin1984}, \(B_m\) is viable under \(\dot{x}=g(x)\) if and only if
	\(g(x)\in T_{B_m}(x)\) for every \(x\in\partial B_m\).  Using the expression for
	the tangent cone, this is equivalent to
	\(\nabla B(x)\cdot g(x)\geq0\)
	for every state satisfying \(B(x)=m\).
\end{proof}

Theorem \ref{thm:viable-breakaway} is the dynamic version of the breakaway condition.  At the boundary of
the breakaway region, the induced motion of the economy must not reduce the
breakaway surplus.  Viability is therefore the persistence condition for a
non-labour saving-investment circuit that generates non-negative continuation
surplus without worker saving.

\subsection{A Linear Benchmark}
\label{subsec:linear-benchmark}

We next consider a linear benchmark problem, in which the cost is linear, with the unit
price of each investment flow normalised to one. The investment cost is
\(C(I;x) = I^K + I^\Lambda + I^a + I^\beta\).
The state dynamics are
\begin{equation}
	\label{eq:linear-dynamics}
	\dot K=I^K-\delta^K K,
	\qquad
	\dot\Lambda=I^\Lambda-\delta^\Lambda\Lambda,
	\qquad
	\dot a=I^a-\delta^a a,
	\qquad
	\dot\beta=\chi I^\beta(1-\beta).
\end{equation}

As $\beta$ measures implemented automation coverage and frontier investment
is non-negative, \(I^\beta \geq 0\) implies that \(\dot\beta \geq 0 \) in this benchmark.
Deployed automation is therefore irreversible along the investment path:
frontier investment can expand the automated task domain but does not undo past
deployment.

In this benchmark, the state dependence of \(\eta^\beta(x)\) enables the continuation
requirement to change as the frontier advances.  In particular, \(\eta^\beta(x)\)
may decline, and may tend to zero as \(\beta\to1\).  One reduced form interpretation
is that the set of investment opportunities outside the automation programme
shrinks as automation expands, lowering the frontier advance required for
continuation.  Hence the investment required to advance the frontier need not
diverge as full automation is approached: its limiting behaviour depends on the
rate at which \(\eta^\beta(x)\) declines relative to \(1-\beta\).

For a target frontier advance \(\eta^\beta(x)>0\), the required frontier investment is
\begin{equation}
	\label{eq:Ibeta-linear}
	I^\beta(x)=\frac{\eta^\beta(x)}{\chi(1-\beta)}.
\end{equation}
The full investment requirement is therefore
\begin{equation}
	\label{eq:tfull-linear}
	\mathcal{T}^{\mathrm{full}}(x)
	=
	\delta^K K + \delta^\Lambda\Lambda + \delta^a a + I^\beta(x).
\end{equation}

Now consider the maintained stocks path: investment keeps \(K,\Lambda,a\) constant
and advances the automation frontier at the target rate \(\eta^\beta(x)\). Along
this path,
\(\dot K=\dot\Lambda=\dot a=0\) and \(\dot\beta=\eta^\beta(x)>0\).
The viability condition at a boundary state \(B(x)=m\) becomes
\(B_\beta(x)\eta^\beta(x)\geq0\).
Since \(\eta^\beta(x)>0\), this is equivalent to
\begin{equation}
	\label{eq:beta-viability-linear}
	B_\beta(x)\geq0.
\end{equation}
Using
\[
	B(x)
	=
	\Sigma^E(x)
	-
	\delta^K K-\delta^\Lambda\Lambda-\delta^a a-I^\beta(x),
\]
we obtain
\begin{equation}
	\label{eq:beta-viability-linear-expanded}
	\Sigma^E_\beta(x)-I^\beta_\beta(x)\geq0.
\end{equation}

This is the marginal viability condition in the maintained-stocks benchmark.  The
automation frontier must generate sufficient additional non-labour saving to cover the
marginal increase in the cost of advancing the frontier.  In this benchmark, the
frontier advance must carry its own financial weight.

\subsection{Marginal Non-Labour Saving in the Production Benchmark}
\label{subsec:marginal-saving-production-example}

The production benchmark makes explicit how frontier advancement affects the
marginal non-labour-saving term in the viability condition.  Recall that
\[
	M
	=
	\left[
		(1-\beta)a^\theta L^\rho+\beta\Lambda^\rho
	\right]^{1/\rho},
	\qquad
	Y=K^\alpha M^{1-\alpha},
\]
with \(\rho\in(0,1)\).  The frontier \(\beta\) shifts weight inside the labour mean
from effective human labour, \(a^{\theta/\rho}L\), towards artificial labour,
\(\Lambda\).

Define the artificial labour share inside the labour mean by
\begin{equation}
	\label{eq:artificial-labour-share}
	q = \frac{\beta\Lambda^\rho}{M^\rho}.
\end{equation}
Competitive factor payments imply
\begin{equation}
	\label{eq:factor-incomes-production-example}
	Y^K=\alpha Y,
	\qquad
	Y^\Lambda=(1-\alpha)qY.
\end{equation}
Hence non-labour saving is
\begin{equation}
	\label{eq:non-labour-saving-production-example}
	\Sigma^E
	=
	\left[
		\sigma^K\alpha+\sigma^\Lambda(1-\alpha)q
	\right]Y.
\end{equation}

Differentiating with respect to \(\beta\), holding \(K,L,\Lambda,a\) fixed, yields
the decomposition
\begin{equation}
	\label{eq:marginal-saving-decomposition}
	\Sigma^E_\beta
	=
	\underbrace{
		\left[
			\sigma^K\alpha+\sigma^\Lambda(1-\alpha)q
		\right]Y_\beta
	}_{\text{output/productivity effect}}
	+
	\underbrace{
		\sigma^\Lambda(1-\alpha)Yq_\beta
	}_{\text{income-composition effect}}.
\end{equation}

The first term captures the effect of frontier advancement on total output.  For
the benchmark production function,
\[
	\frac{Y_\beta}{Y}
	=
	\frac{1-\alpha}{\rho}
	\frac{\Lambda^\rho-a^\theta L^\rho}{M^\rho}.
\]
Since \(\rho\in(0,1)\),
\begin{equation}
	\label{eq:output-effect-threshold}
	Y_\beta>0
	\quad\Longleftrightarrow\quad
	\Lambda>a^{\theta/\rho}L.
\end{equation}
Hence frontier advancement raises output only when artificial labour is sufficiently
productive relative to effective human labour in the remaining human-essential
task domain.  Below this threshold, marginal automation lowers output.

The second term captures the induced shift of income toward artificial labour.  For
the benchmark production function,
\begin{equation}
	\label{eq:q-beta}
	q_\beta
	=
	\frac{\Lambda^\rho a^\theta L^\rho}{M^{2\rho}}
	=
	\frac{q(1-q)}{\beta(1-\beta)}.
\end{equation}
Using the expression for labour income,
\[
	Y^L
	=
	(1-\alpha)Y
	\frac{(1-\beta)a^\theta L^\rho}{M^\rho},
\]
the income-composition effect can be written as
\begin{equation}
	\label{eq:composition-effect-labour-income-density}
	\sigma^\Lambda(1-\alpha)Yq_\beta
	=
	\sigma^\Lambda
	\frac{q}{\beta}
	\frac{Y^L}{1-\beta}.
\end{equation}
Thus the composition effect is tied to labour income per unit of remaining
human-essential task margin, not to aggregate labour income itself.

Substituting the above into Equation \eqref{eq:marginal-saving-decomposition} yields
\begin{equation}
	\label{eq:marginal-saving-beta-explicit}
	\Sigma^E_\beta
	=
	Y\left[
	\left(\sigma^K\alpha+\sigma^\Lambda(1-\alpha)q\right)
	\frac{1-\alpha}{\rho}
	\frac{\Lambda^\rho-a^\theta L^\rho}{M^\rho}
	+
	\sigma^\Lambda(1-\alpha)
	\frac{\Lambda^\rho a^\theta L^\rho}{M^{2\rho}}
	\right].
\end{equation}

The viability condition in the maintained stocks linear benchmark can therefore be
written as
\begin{equation}
	\label{eq:viability-decomposition}
	\underbrace{
		\left[
			\sigma^K\alpha+\sigma^\Lambda(1-\alpha)q
		\right]Y_\beta
	}_{\text{output/productivity effect}}
	+
	\underbrace{
		\sigma^\Lambda(1-\alpha)Yq_\beta
	}_{\text{income-composition effect}}
	\geq
	\underbrace{
		I^\beta_\beta(x)
	}_{\text{marginal frontier cost}}.
\end{equation}
The second term is a pure income-source composition effect.  It does not require
wealth dependent saving or different saving rates between physical capital and
artificial labour.  It arises because frontier advancement shifts factor income
toward artificial labour, which belongs to the accumulable bloc.

Equivalently,
\begin{equation}
	\label{eq:linear-viability-production-example}
	Y\left[
	\left(\sigma^K\alpha+\sigma^\Lambda(1-\alpha)q\right)
	\frac{1-\alpha}{\rho}
	\frac{\Lambda^\rho-a^\theta L^\rho}{M^\rho}
	+
	\sigma^\Lambda(1-\alpha)
	\frac{\Lambda^\rho a^\theta L^\rho}{M^{2\rho}}
	\right]
	\geq
	I^\beta_\beta(x).
\end{equation}

\begin{corollary}
\label{cor:breakaway-negative-productivity}
Consider a boundary state \(B(x)=m\) in the maintained-stocks linear benchmark and
suppose that \(Y_\beta(x)<0\).  The maintained-stocks path satisfies the viability
condition at \(x\) if and only if
\begin{equation}
\label{eq:negative-productivity-breakaway}
\sigma^\Lambda(1-\alpha)Yq_\beta
\geq
I^\beta_\beta(x)
+
\left[\sigma^K\alpha+\sigma^\Lambda(1-\alpha)q\right](-Y_\beta).
\end{equation}
Thus frontier automation can support viable breakaway even when it lowers output:
the income-composition effect must cover both the marginal frontier cost and the
associated loss of non-labour saving through lower output.
\end{corollary}

\begin{proof}
With \(Y_\beta<0\), rearranging Equation~\eqref{eq:viability-decomposition} gives
Equation~\eqref{eq:negative-productivity-breakaway}.
\end{proof}

In general,
\begin{equation}
	\label{eq:Ibeta-beta-general}
	I^\beta_\beta(x)
	=
	\frac{\eta^\beta_\beta(x)}{\chi(1-\beta)}
	+
	\frac{\eta^\beta(x)}{\chi(1-\beta)^2}.
\end{equation}
If \(\eta^\beta\) is locally independent of \(\beta\), then
\begin{equation}
	\label{eq:Ibeta-beta-constant-eta}
	I^\beta_\beta(x)
	=
	\frac{\eta^\beta}{\chi(1-\beta)^2}.
\end{equation}

This is an important point: viable breakaway is an investment condition, not a productivity condition.  In
the production benchmark, marginal automation supports breakaway when the
output/productivity and income-composition effects generate enough additional
non-labour saving to cover the marginal cost of advancing the frontier.  This can
happen through higher output, through a shift of income toward artificial labour,
or through both.  The central object is therefore marginal non-labour saving
relative to the marginal cost of continuation.

\section{Labour Income, Tasks, and Jobs}
\label{sec:labour-income-jobs}

Breakaway is an investment-surplus condition: it identifies when accumulation no
longer requires worker saving.  Labour income matters to this question for two
reasons.  It measures the economic return to human work, but before breakaway it
also generates the worker-saving flow \(\sigma^L Y^L\) on which continued
accumulation may still depend.  The labour-side consequences of automation
therefore also have a financing dimension: they determine the fate of one of the
economy's financing sources.  This section studies paths along which the automation
frontier approaches full coverage, \(\beta(t)\to1\), whatever the reason for that
convergence.

We first ask what happens to labour income along such paths.  The key issue is the
cumulative balance among the shrinking human-essential task margin, \(1-\beta(t)\),
the augmentation of remaining human labour through augmentation capability
\(a(t)\), and productive-environment effects generated by the evolution of the
other inputs.
We then extend the analysis from tasks to jobs.  Since jobs are bundles of tasks,
the same automation frontier can imply different employment outcomes depending
on how core human requirements are distributed across jobs.

\subsection{Labour Income Near Full Automation}
\label{subsec:labour-income-near-full-automation}

Consider a full-automation path, namely a trajectory \(x(t)\) along which
\(\beta(t)\equiv\beta(x(t))\to1\).  The question is what happens to labour income
as the human-essential task margin disappears.  The production benchmark provides
a useful foundation for the relevant asymptotic form.  Recall that in our leading
example
\[
	M
	=
	\left[
		(1-\beta)a^\theta L^\rho+\beta\Lambda^\rho
	\right]^{1/\rho},
	\qquad
	Y=K^\alpha M^{1-\alpha}.
\]
This combines a Cobb--Douglas production function in physical capital and the
labour mean with a constant-elasticity mean of human and artificial labour
services.  The specification preserves the standard capital--labour structure
while allowing the labour input itself to shift from human labour toward
artificial labour as the automation frontier advances.

For this benchmark, labour income is
\[
	Y^L
	=
	(1-\alpha)Y
	\frac{(1-\beta)a^\theta L^\rho}{M^\rho}.
\]
Along displacement dominated full-automation paths in the benchmark, and more
generally whenever the artificial labour term dominates inside the labour mean,
\(M^\rho\sim\Lambda^\rho\).  Hence
\begin{equation}
	\label{eq:labour-income-benchmark-asymptotic}
	Y^L(x(t))
	\sim
	(1-\beta(t))a(t)^\theta
	(1-\alpha)K(t)^\alpha L(t)^\rho
	\Lambda(t)^{1-\alpha-\rho}.
\end{equation}
This motivates the more general asymptotic representation
\begin{equation}
	\label{eq:labour-income-asymptotic}
	Y^L(x(t))
	\sim
	(1-\beta(t))a(t)^\theta h(K(t),L(t),\Lambda(t)).
\end{equation}
The term \(1-\beta(t)\) captures displacement: the contraction of the
human-essential task domain.  The term \(a(t)^\theta\) captures augmentation: the
effect of augmentation capability on human labour in that domain.

Define the human-labour margin
\begin{equation}
\label{eq:human-labour-margin}
H(t)=(1-\beta(t))a(t)^\theta .
\end{equation}

The term \(H(t)\) isolates the direct balance between displacement, through the
contraction of the human-essential task margin \(1-\beta(t)\), and augmentation,
through growth in augmentation capability \(a(t)\). The remaining term
\(h(K(t),L(t),\Lambda(t))\) captures the effects of the surrounding productive
environment.

When \(a(t)>0\), \(h(K(t),L(t),\Lambda(t))>0\), and the path is differentiable,
define
\[
\rho_\beta(t)
=
-\frac{d}{dt}\log(1-\beta(t)),
\qquad
\rho_a(t)
=
\frac{d}{dt}\log a(t),
\]
and
\[
\rho_h(t)
=
\frac{d}{dt}
\log h(K(t),L(t),\Lambda(t)).
\]

\begin{proposition}
\label{prop:labour-income-collapse-persistence}
Consider a full-automation path \(x(t)\) with \(\beta(t)\to1\). Suppose that
labour income satisfies Equation~\eqref{eq:labour-income-asymptotic}, that
\(a(t)>0\) and \(h(K(t),L(t),\Lambda(t))>0\) along the path, and that
the path is differentiable. Then labour income collapses if and only if
\begin{equation}
\label{eq:labour-income-collapse-growth}
\int_0^t
\left[
\rho_\beta(s)
-\theta\rho_a(s)
-\rho_h(s)
\right]ds
\to+\infty.
\end{equation}
Moreover, labour income is eventually bounded away from zero, in the sense that
\(\liminf_{t\to\infty}Y^L(x(t))>0\), if and only if
\begin{equation}
\label{eq:labour-income-persistence-growth}
\limsup_{t\to\infty}
\int_0^t
\left[
\rho_\beta(s)
-\theta\rho_a(s)
-\rho_h(s)
\right]ds
<+\infty.
\end{equation}
\end{proposition}

\begin{proof}[Proof sketch]
\renewcommand{\qedsymbol}{}
Equation \eqref{eq:labour-income-asymptotic} makes labour income asymptotically proportional to the product of the human-labour margin and the productive-environment term. Writing this product in exponential form makes its evolution depend on the cumulative gap between displacement, augmentation, and the productive environment. The two equivalences then follow from the limiting behaviour of that cumulative gap. A full proof is given in Appendix~\ref{app:proof-labour-income-collapse}.
\end{proof}

Proposition \ref{prop:labour-income-collapse-persistence} shows that labour income is governed by the cumulative balance
between three forces: displacement, augmentation, and changes in the surrounding
productive environment.
The contraction of the human-essential task margin contributes \(\rho_\beta\)
toward labour income collapse, while augmentation contributes \(\theta\rho_a\)
and changes in the remaining productive environment contribute \(\rho_h\) toward
labour income persistence.

In the production benchmark, where
\[h(K,L,\Lambda)=(1-\alpha)K^\alpha L^\rho\Lambda^{1-\alpha-\rho}\]
the productive-environment term has growth rate
\[
\rho_h(t)
=
\alpha\rho_K(t)
+
\rho\rho_L(t)
+
(1-\alpha-\rho)\rho_\Lambda(t),
\]
where
\[
\rho_K(t)=\frac{\dot K(t)}{K(t)},
\qquad
\rho_L(t)=\frac{\dot L(t)}{L(t)},
\qquad
\rho_\Lambda(t)=\frac{\dot\Lambda(t)}{\Lambda(t)}.
\]
If human-labour supply \(L\) is constant along the path, this reduces to
\(\rho_h(t)=\alpha\rho_K(t)+(1-\alpha-\rho)\rho_\Lambda(t)\).
Hence, in this benchmark, labour income collapses if and only if
\[
\int_0^t
\left[
\rho_\beta(s)
-\theta\rho_a(s)
-\alpha\rho_K(s)
-(1-\alpha-\rho)\rho_\Lambda(s)
\right]ds
\to+\infty.
\]

Frontier expansion therefore need not eliminate labour income even when
\(\beta(t)\to1\): sufficiently strong augmentation or capital deepening can
offset the shrinking human task margin, while the effect of artificial-labour
accumulation depends on \(1-\alpha-\rho\), as shown below.

The proposition identifies the aggregate labour-income margin relevant for the
investment-surplus problem. Worker saving is generated by labour income:
\(\Sigma(x)-\Sigma^E(x)=\sigma^L Y^L(x)\).
Hence, along full-automation paths satisfying
Equation~\eqref{eq:labour-income-collapse-growth},
\(\Sigma(x(t))-\Sigma^E(x(t))\to0\).
The worker-saving margin then disappears endogenously. The decomposition of
labour-income collapse into job and intensive margins requires a mapping from
the automation frontier to human-filled jobs. We introduce that mapping next.

\subsection{Tasks, Jobs, and Labour Income}
\label{subsec:tasks-jobs-labour-income}

The automation frontier is a task-level object, while jobs bundle tasks at a
higher level.  To connect the two, let \(\tau\in\{1,2,\ldots\}\) denote the
number of core tasks in a job: these are the tasks that must all be automated for
the job to be fully automated.  Let \(p_\tau\) denote the
share of jobs with exactly \(\tau\) core tasks, with
\(\sum_{\tau\geq1}p_\tau=1\).
A job remains human-filled whenever at least one of its core tasks remains outside
the automated task domain.  In the benchmark case in which core tasks are distributed
independently over the task space, a job with \(\tau\) core tasks is fully automated
with probability \(\beta^\tau\).  Hence the share of jobs that remain human-filled is
\begin{equation}
	\label{eq:human-filled-jobs}
	J^H(\beta)
	=
	\sum_{\tau\geq1}p_\tau(1-\beta^\tau).
\end{equation}

\begin{proposition}
	\label{prop:task-automation-jobs}
	For any distribution \(\{p_\tau\}_{\tau\geq1}\),
	\(J^H(\beta)\to0\) as \(\beta\to1\).
	If
	\(\bar\tau=\sum_{\tau\geq1}\tau p_\tau<\infty\),
	then
	\begin{equation}
		\label{eq:jobs-asymptotic}
		J^H(\beta)
		\sim
		\bar\tau(1-\beta)
		\qquad\text{as }\beta\to1.
	\end{equation}
\end{proposition}

The proof is given in Appendix~\ref{app:proof-task-automation-jobs}.

The proposition translates frontier automation into the extensive margin of
human-filled jobs.  Let \(\bar N\) denote the mass of potential jobs.  The mass of
human-filled jobs is
\begin{equation}
	\label{eq:human-filled-job-mass}
	N^H(\beta)=\bar N J^H(\beta).
\end{equation}
Aggregate labour income \(Y^L(x)\) is determined by the production side.  Whenever
\(J^H(\beta)>0\), define \(\omega(x)\) as average labour income per human-filled
job:
\begin{equation}
	\label{eq:average-income-human-jobs}
	\omega(x)
	=
	\frac{Y^L(x)}{N^H(\beta)}
	=
	\frac{Y^L(x)}{\bar N J^H(\beta)}.
\end{equation}
Thus \(\omega\) is an intensive labour-income margin at the job level: the average
labour income paid across the jobs that remain human-filled. The two intensive
margins have distinct interpretations: \(w\) is the return per unit of human labour,
while \(\omega\) is aggregate labour income per remaining human-filled job.

Along a full-automation path satisfying
\[
	Y^L(x(t))
	\sim
	(1-\beta(t))a(t)^\theta h(K(t),L(t),\Lambda(t)),
\]
and with \(\bar\tau<\infty\), Proposition~\ref{prop:task-automation-jobs} implies
\(J^H(\beta(t))\sim\bar\tau(1-\beta(t))\).
Hence the average labour income per human-filled job satisfies
\begin{equation}
	\label{eq:income-human-jobs-asymptotic}
	\omega(x(t))
	\sim
	\frac{
		a(t)^\theta h(K(t),L(t),\Lambda(t))
	}{
		\bar N\bar\tau
	}.
\end{equation}

This cancellation is the key labour-market implication of the finite-mean task
benchmark.  Near full automation, the frontier term \(1-\beta\) contracts aggregate
labour income and the measure of human-filled jobs at the same rate.  It therefore
drops out of average labour income per human-filled job.  Conditional on the jobs
that remain human-filled, the intensive margin is governed by automation
capability and the scale term \(h\).

Labour-income collapse can therefore be an extensive-margin phenomenon.  For
example, if
\(a(t)^\theta h(K(t),L(t),\Lambda(t))\to c\in(0,\infty)\),
then aggregate labour income vanishes as \(\beta(t)\to1\), while
\(\omega(x(t))\to c/(\bar N\bar\tau)>0\).
\begin{corollary}
\label{cor:diverging-income-human-jobs}
Consider a full-automation path with \(\bar\tau<\infty\).  If
\[a(t)^\theta h(K(t),L(t),\Lambda(t))\to\infty\]
while
\[(1-\beta(t))a(t)^\theta h(K(t),L(t),\Lambda(t))\to0,\]
then
\[
	Y^L(x(t))\to0,
	\qquad
	N^H(\beta(t))\to0,
	\qquad and\;\;\; 
	\omega(x(t))\to\infty.
\]
\end{corollary}

\begin{proof}
The second condition and Equation~\eqref{eq:labour-income-asymptotic} imply
\(Y^L(x(t))\to0\).  Proposition~\ref{prop:task-automation-jobs} implies
\(N^H(\beta(t))\to0\), while Equation~\eqref{eq:income-human-jobs-asymptotic}
and the first condition imply \(\omega(x(t))\to\infty\).
\end{proof}

Since \(w(x)=Y^L(x)/L\), the corresponding wage implication depends on the
path of human-labour supply. If \(L\) remains bounded away from zero---in
particular, if labour supply is constant---labour income collapse implies
\(w(x(t)) \to 0\), even as average labour income per surviving human-filled
job may remain positive or diverge.

This possibility is related in spirit to \citet{raymookherjee2022growth}, where
progressive automation can drive the labour share to zero while absolute real wages
rise without bound; here the result instead follows from the shrinking extensive
margin of human-filled jobs.

The task-to-job structure determines the extensive job margin; automation
capability and the scale term determine the intensive labour-income margin within
the jobs that remain.

For the investment-surplus problem, the aggregate labour-income margin remains
decisive.
Worker saving depends on \(Y^L\), not separately on the number of human-filled jobs
or on average labour income within those jobs:
\(\Sigma(x)-\Sigma^E(x)=\sigma^L Y^L(x)\).
Thus, when labour income collapses, the worker-saving margin disappears even if
some human-filled jobs remain and even if average labour income in those jobs does
not collapse.  This is the object that the optimal-reliance problem takes as given: the
automation economy may retain human-filled jobs while ceasing to draw on worker
saving as a source of automation finance.

\subsection{Artificial-Labour Accumulation and Financing Alignment}
\label{subsec:artificial-labour-financing}

Artificial labour affects both sides of the financing comparison.  Its income
contributes directly to non-labour saving, while its accumulation may either raise
or reduce the labour income that generates worker saving.  The production
benchmark gives a sharp characterisation of these two effects.

\begin{proposition}
\label{prop:artificial-labour-accumulation}
In the production benchmark, at an interior state and holding
\(K,L,a,\beta\) fixed,
\begin{equation}
\label{eq:SigmaE-lambda}
\frac{\partial\Sigma^E}{\partial\log\Lambda}
=
(1-\alpha)qY
\left[
\sigma^K\alpha
+\sigma^\Lambda(1-\alpha)q
+\sigma^\Lambda\rho(1-q)
\right]
\geq0.
\end{equation}
Since \(q\in(0,1)\) at an interior state and \(\rho\in(0,1)\), the
derivative is strictly positive whenever \(\Sigma^E>0\).  By contrast,
\begin{equation}
\label{eq:YL-lambda}
\frac{\partial\log Y^L}{\partial\log\Lambda}
=
(1-\alpha-\rho)q.
\end{equation}
Hence artificial-labour accumulation raises labour income if
\(\rho<1-\alpha\), leaves it unchanged if \(\rho=1-\alpha\), and lowers it if
\(\rho>1-\alpha\).  Since the elasticity of substitution between human and
artificial labour is \(1/(1-\rho)\), the threshold is equivalently an elasticity
of substitution equal to \(1/\alpha\).

In the linear investment benchmark,
\begin{equation}
\label{eq:B-lambda}
\frac{\partial B}{\partial\log\Lambda}
=
\frac{\partial\Sigma^E}{\partial\log\Lambda}
-\delta^\Lambda\Lambda
-\Lambda I^\beta_\Lambda(x).
\end{equation}
Thus a marginal increase in the artificial-labour stock moves the state toward
breakaway if and only if
\begin{equation}
\label{eq:B-lambda-condition}
(1-\alpha)qY
\left[
\sigma^K\alpha
+\sigma^\Lambda(1-\alpha)q
+\sigma^\Lambda\rho(1-q)
\right]
\geq
\delta^\Lambda\Lambda+\Lambda I^\beta_\Lambda(x).
\end{equation}
If the target frontier rate is locally independent of \(\Lambda\), this reduces to
the condition that the marginal non-labour saving generated by a larger
artificial-labour stock covers its marginal maintenance requirement,
\begin{equation}
\label{eq:B-lambda-condition-simple}
(1-\alpha)qY
\left[
\sigma^K\alpha
+\sigma^\Lambda(1-\alpha)q
+\sigma^\Lambda\rho(1-q)
\right]
\geq
\delta^\Lambda\Lambda.
\end{equation}
\end{proposition}

The proof is given in Appendix~\ref{app:proof-artificial-labour-accumulation}.

The result of Proposition \ref{prop:artificial-labour-accumulation} separates two effects of artificial labour accumulation.  A larger value of
\(\Lambda\) always expands the non-labour saving base: both the output effect and
the shift toward artificial-labour income are positive.  Its effect on the
worker-saving base, however, depends on substitution.  At states satisfying
Equation~\eqref{eq:B-lambda-condition}, if \(\rho>1-\alpha\), a larger
artificial-labour stock moves the economy toward breakaway while simultaneously
eroding labour income; financial independence and erosion of worker finance are
therefore aligned.  If \(\rho<1-\alpha\), the same accumulation can move the
economy toward breakaway while raising labour income; financial independence then
becomes more feasible even as the worker-financing source remains valuable.  This
distinction will reappear in the optimal-reliance problem below.

\section{Optimal Reliance on Worker Finance}
\label{sec:optimal}

\subsection{The Optimal-Reliance Problem}
\label{subsec:optimal-reliance-control-problem}

Breakaway and viability study whether non-labour saving can finance continuation,
while the preceding section characterises the worker-saving base,
\(\sigma^L Y^L(x)\).  We now extend this question from viability to optimality by
making the path of reliance on worker finance endogenous.  The investment
programme may draw on the worker-saving base \(\sigma^L Y^L(x)\), but doing so may
carry an external-finance wedge, while reducing that reliance may itself be costly.
We therefore ask how much worker finance an optimal automation programme uses
along the transition and how its reliance evolves over time.  To isolate this
optimality question from financing feasibility, we treat any otherwise admissible
investment path as financeable.

Denote by 
\(x(t)=(K(t),\Lambda(t),a(t),\beta(t))\)
the automation state, and by
\[I(t)=(I^K(t),I^\Lambda(t),I^a(t),I^\beta(t))\]
 the investment vector.  The automation state evolves according to
\(\dot{x}(t)=G(x(t),I(t))\).
As in the definition of \(\mathcal{T}^{\mathrm{full}}\), the investment path must
maintain the automation system and advance the frontier:
\[
	G^K(x(t),I(t))\geq0,
	\qquad
	G^\Lambda(x(t),I(t))\geq0,
	\qquad
	G^a(x(t),I(t))\geq0,
	\qquad
	G^\beta(x(t),I(t))\geq\eta^\beta(x(t)).
\]
These continuation requirements keep the technological programme comparable to
the programme used to define the full-investment requirement.

The new state variable in this section is reliance on worker finance.  Let
\(u(t)\in[0,1]\).
We distinguish the source of funds from its use in the accumulation circuit:
\(\sigma^L Y^L(x)\) is worker saving, while \(u\sigma^L Y^L(x)\) is the worker
finance supplied to the automation programme.  Thus \(u\) measures the degree of
reliance on worker finance.  When \(u=1\), the investment programme uses all worker
saving available to it; when \(u=0\), it no longer draws on worker saving.  At
state \(x\), the saving flow available to the accumulation circuit is
\(\Sigma^E(x)+u\sigma^L Y^L(x)\).
Worker finance is therefore a financing input into the optimal investment
programme, and the programme determines the extent and duration of its reliance
through the path of the reliance reduction rate defined next.

Let the reliance reduction rate satisfy
\(z(t)\geq0\)
and let reliance evolve according to
\(\dot{u}(t)=-z(t)\).
A high value of \(z(t)\) means a rapid reduction in reliance; a low value means a
gradual reduction.  The restriction \(z(t)\geq0\) defines the optimality problem
over the class of non-increasing reliance paths.  Once \(u(t)=0\), reliance has
reached zero and feasibility requires \(z(t)=0\).

The key object is the investor's surplus.  For a given investment programme, the
\emph{flow surplus} is the saving available to the accumulation circuit minus the cost of
implementing the programme.  In the transition problem, the cost of a programme may
depend not only on the automation investment vector \(I\), but also on the current
degree of reliance on worker finance and on the speed at which this reliance is
being reduced.  We therefore write the transition cost as
\(C(I,u,z;x)\),
and the corresponding flow investment surplus is
\(\Sigma^E(x)+u\sigma^L Y^L(x)-C(I,u,z;x)\).

Worker saving has a direct financing benefit: it raises the saving flow available
to the investor by \(u\sigma^L Y^L(x)\).  Reliance on worker finance may also carry
an external-finance wedge relative to internally generated non-labour finance,
arising for example from intermediation, contracting, information, or control
costs.  We capture this wedge in reduced form through the dependence of
\(C(I,u,z;x)\) on \(u\).  We assume that, on interior reliance states,
\(C_u(I,u,z;x)>0\),
so that greater reliance on worker finance raises the programme cost, other things
equal.  The optimal-reliance results below are conditional on this positive
marginal external-finance wedge.  Since \(u\) scales the share of worker saving
supplied as worker finance, both \(C_u\) and \(\sigma^L Y^L(x)\) are marginal flow
quantities.  Their difference determines the current net value of worker finance,
while the discount rate \(r\) determines how future values of this difference are
weighted.  The link with the preceding section enters through the state \(x\), and
in particular through labour income \(Y^L(x)\).  As the labour-income base shrinks,
the amount of finance generated by worker saving falls.  If the marginal
external-finance wedge does not fall at the same rate, continued reliance on worker
finance becomes progressively less attractive.

The dependence of \(C\) on \(z\) captures the cost of reducing reliance quickly.  We
assume
\[
	C_z(I,u,0;x)=0,
	\qquad
	C_z(I,u,z;x)>0\ \text{for }z>0,
	\qquad
	C_{zz}(I,u,z;x)>0.
\]
Thus faster reliance reduction is more costly at the margin.  The cost function
therefore contains both the financing cost of remaining reliant on worker finance
and the cost of reducing that reliance quickly.

This formulation allows these components to interact through the state \(x\).  In
applications, however, it may be useful to decompose the transition cost as
\(C(I,u,z;x)=C^A(I;x)+C^R(u;x)+C^Z(z;x)\),
where \(C^A\) is the resource cost of automation investment, \(C^R\) is the
financing cost associated with reliance on worker finance, and \(C^Z\) is the
reliance reduction cost.

\begin{definition}
\label{def:optimal-reliance-problem}
The \emph{optimal reliance problem} is to choose the investment path \(I(t)\) and the
reliance reduction rate \(z(t)\) so as to maximise discounted investor surplus,
where \(u(t)\) denotes reliance on worker finance:
\[
	\max_{I(t),z(t)}
	\int_0^\infty e^{-rt}
	\left[
		\Sigma^E(x(t))
		+
		u(t)\sigma^L Y^L(x(t))
		-
		C(I(t),u(t),z(t);x(t))
	\right]dt
\]
subject to
\[
	\dot{x}(t)=G(x(t),I(t)),
	\qquad
	\dot{u}(t)=-z(t),
\]
\[
	G^K(x(t),I(t))\geq0,
	\qquad
	G^\Lambda(x(t),I(t))\geq0,
	\qquad
	G^a(x(t),I(t))\geq0,
	\qquad
	G^\beta(x(t),I(t))\geq\eta^\beta(x(t)),
\]
\[
	I(t)\geq0,
	\qquad
	z(t)\geq0,
	\qquad
	0\leq u(t)\leq1,
\]
and initial conditions
\(x(0)=x_0\) and \(u(0)=u_0\).
\end{definition}

The accounting of Definition~\ref{def:optimal-reliance-problem} is linked to the
earlier breakaway object.  When reliance has reached zero and there is no active
reliance reduction, \(u=0\) and \(z=0\), so the flow surplus is
\(\Sigma^E(x)-C(I,0,0;x)\).
For the same continuation requirements, the full-investment requirement is
\(\mathcal{T}^{\mathrm{full}}(x)=\inf_I C(I,0,0;x)\)
subject to the full-investment continuation requirements.  Recall that the
breakaway surplus is
\(B(x)=\Sigma^E(x)-\mathcal{T}^{\mathrm{full}}(x)\),
which yields the following correspondence without requiring the infimum to be
attained.

\begin{lemma}
\label{lem:breakaway-correspondence}
At any state $x$ at which the full-investment continuation requirements are
feasible,
\[
	\sup_I\left\{\Sigma^E(x)-C(I,0,0;x)\right\}
	=
	\Sigma^E(x)-\mathcal{T}^{\mathrm{full}}(x)
	=
	B(x),
\]
where the supremum is over investment vectors satisfying those continuation
requirements.
\end{lemma}

\begin{proof}
The term \(\Sigma^E(x)\) is independent of \(I\).  Taking the supremum of the flow
surplus therefore subtracts from \(\Sigma^E(x)\) the infimum of \(C(I,0,0;x)\)
over the same feasible set, which is \(\mathcal{T}^{\mathrm{full}}(x)\) by
definition.  The final equality is the definition of \(B(x)\).
\end{proof}

\begin{remark}
\label{rem:breakaway-zero-reliance}
Lemma \ref{lem:breakaway-correspondence} is an accounting correspondence, while
admissibility in Definition \ref{def:optimal-reliance-problem} is defined by the
permissive financing benchmark above.  Breakaway, \(B(x) \geq 0\), is the current-flow
self-financing benchmark studied earlier, whereas optimal \(u=0\) concerns the use
of worker finance when that financing restriction is removed.  The two definitions
answer distinct questions and, by themselves, establish neither implication between
\(B(x)\geq0\) and optimal \(u=0\).
\end{remark}

The optimal path determines whether and when \(u(t)\) reaches zero.  Worker saving
contributes to investment surplus while \(u>0\), while reliance on this external
source carries the reduced-form financing wedge captured by \(C_u\).  A faster
reduction in reliance raises costs through \(C_z\).  Investment changes the
automation state, and the automation state changes both the future contribution of
worker finance and the future cost of relying on it.  The optimal path therefore
jointly determines automation investment and the path of reliance on worker finance.

The remainder of the section first derives the explicit optimality conditions,
then characterises the forward-looking dynamics of optimal reliance and their
implications, and finally gives stronger conditions under which reliance reaches
zero in finite time.

\subsection{Optimality Conditions}
\label{subsec:reliance-optimality-conditions}
Let
\[
	R(x,I)
	=
	\left(
		G^K(x,I),
		G^\Lambda(x,I),
		G^a(x,I),
		G^\beta(x,I)-\eta^\beta(x)
	\right)
\]
denote the vector of continuation requirements, so that the constraints are
\(R(x,I)\geq0\).  Let \(\phi(t)\geq0\) denote the associated current-value
multipliers.  Let \(\lambda(t)\) denote the current-value costate vector associated
with the automation state \(x(t)\), and let \(\mu(t)\) denote the current-value
shadow value of reducing reliance on worker finance.  Equivalently, the standard
costate associated with \(u(t)\) is \(-\mu(t)\).  Thus a positive \(\mu(t)\) means
that an additional unit of reliance has negative marginal dynamic value.

With this sign convention, the current-value Hamiltonian is
\[
	\mathcal{H}
	=
	\Sigma^E(x)
	+
	u\sigma^L Y^L(x)
	-
	C(I,u,z;x)
	+
	\lambda\cdot G(x,I)
	+
	\mu z
	+
	\phi\cdot R(x,I).
\]

\begin{lemma}
\label{lem:interior-optimality-conditions}
Let \((I,z)\) be an optimal control path and consider an arc on which \(u(t)>0\), so
that the state constraint \(u(t)\geq0\) is not binding.  Suppose there exist
current-value multipliers \((\lambda,-\mu,\phi)\) satisfying the current-value
maximum principle for Definition~\ref{def:optimal-reliance-problem}: Hamiltonian
maximisation over the admissible controls, the current-value costate equations, and
complementary slackness for the continuation requirements.  Then
\[
	\phi(t)\geq0,
	\qquad
	R(x(t),I(t))\geq0,
	\qquad
	\phi(t)\cdot R(x(t),I(t))=0.
\]
For an interior investment choice, the first-order condition is
\[
	C_I(I,u,z;x)
	=
	G_I(x,I)'\lambda
	+
	R_I(x,I)'\phi.
\]
For the reliance-reduction choice,
\[
	z(t)\geq0,
	\qquad
	C_z(I(t),u(t),z(t);x(t))-\mu(t)\geq0,
	\qquad
	z(t)\left[
		C_z(I(t),u(t),z(t);x(t))-\mu(t)
	\right]=0.
\]
Hence, on an interior arc with declining reliance, where \(z(t)>0\),
\(C_z(I,u,z;x)=\mu\).
The shadow equation for reliance on worker finance is
\[
	\dot{\mu}
	=
	r\mu
	+
	\sigma^L Y^L(x)
	-
	C_u(I,u,z;x).
\]
The costate equation for the automation state is
\[
	\dot{\lambda}
	=
	r\lambda
	-
	\Sigma^E_x(x)
	-
	u\sigma^L Y^L_x(x)
	+
	C_x(I,u,z;x)
	-
	G_x(x,I)'\lambda
	-
	R_x(x,I)'\phi .
\]
\end{lemma}

\begin{proof}
Write the abstract maximum-principle conditions for the Hamiltonian above.  For an
interior investment choice, Hamiltonian maximisation gives the stated investment
first-order condition.  Maximisation with the constraint \(z\geq0\) gives the stated
KKT conditions and therefore \(C_z=\mu\) whenever \(z>0\).  Since the standard
costate of \(u\) is \(-\mu\), its current-value costate equation, together with the
\(u\)-derivative of the Hamiltonian, gives the stated shadow equation.  Applying the
current-value costate equation to \(x\) gives the stated equation for \(\lambda\).
Complementary slackness gives the displayed conditions for \(\phi\).
\end{proof}

Investment is chosen such that the marginal programme cost of an interior investment
component is balanced against the dynamic value of changing the automation state
and the value of satisfying the continuation requirements.  When a continuation
requirement binds, its multiplier raises the marginal value of investment in the
corresponding direction.  On an arc with declining reliance, the marginal cost of
reducing reliance more rapidly equals the shadow value of reducing reliance on
worker finance.

\begin{corollary}
\label{cor:reliance-reduction-euler}
On an interior arc with declining reliance,
\[
	\frac{d}{dt}C_z(I(t),u(t),z(t);x(t))
	=
	rC_z(I(t),u(t),z(t);x(t))
	+
	\sigma^L Y^L(x(t))
	-
	C_u(I(t),u(t),z(t);x(t)).
\]
\end{corollary}

\begin{proof}
On such an arc, Lemma~\ref{lem:interior-optimality-conditions} gives
\(C_z=\mu\).  Differentiating this identity along the path and substituting the
shadow equation for \(\mu\) gives the result.
\end{proof}

The state \(x\) affects non-labour saving, labour income, programme costs, and the
continuation constraints.  Investment therefore changes not only the future
automation state, but also the future benefit and cost of remaining reliant on
worker finance.  The zero-reliance accounting is given by
Lemma \ref{lem:breakaway-correspondence}; as emphasised in
Remark \ref{rem:breakaway-zero-reliance}, it is an accounting correspondence rather
than an equivalence between breakaway and zero reliance.

\subsection{Optimal Reliance Dynamics}
\label{subsec:optimal-reliance-dynamics}

The central implication of the optimality conditions is that the value of reducing
reliance on worker finance is inherently forward-looking.  The next proposition
collects the forward representation of that value and the threshold rule governing
whether reliance declines.

\begin{proposition}
\label{prop:forward-looking-reliance}
Suppose the conditions of Lemma~\ref{lem:interior-optimality-conditions} hold.
For the forward representation, suppose additionally that, from time \(t\) onward,
the continuation path remains interior, so that \(u(s)>0\) for all \(s\geq t\),
and that
\[
	\lim_{s\to\infty}e^{-r(s-t)}\mu(s)=0.
\]
Then
\[
	\mu(t)
	=
	\int_t^\infty
	e^{-r(s-t)}
	\left[
		C_u(I(s),u(s),z(s);x(s))
		-
		\sigma^L Y^L(x(s))
	\right]ds.
\]
Independently of these additional interiority and terminal conditions, the KKT conditions, together with
\(C_z(I,u,0;x)=0\) and \(C_z>0\) for \(z>0\), imply that reliance is constant
when
\(\mu(t)\leq0\),
and reliance is declining when
\(\mu(t)>0\).
On an arc of declining reliance, the reliance reduction rate is determined by
\(C_z(I,u,z;x)=\mu(t)\).
\end{proposition}

\begin{proof}
Use the shadow equation in Lemma~\ref{lem:interior-optimality-conditions}.  Multiplying
it by the current-value integrating factor and integrating forward from \(t\) to a
finite horizon gives the present value of the net marginal cost of reliance plus a
terminal term.  The stated terminal condition makes that term vanish as the horizon
tends to infinity, yielding the forward representation.  Independently of the
additional interiority and terminal conditions, the threshold rule follows from the
KKT conditions in Lemma~\ref{lem:interior-optimality-conditions} together with
\(C_z(I,u,0;x)=0\) and \(C_z>0\) for \(z>0\).
\end{proof}

$\mu(t)$ is the present value of the future net marginal cost of reliance on
worker finance.  The term \(C_u\) is the marginal external-finance wedge associated
with keeping worker finance inside the accumulation circuit.  The term
\(\sigma^L Y^L\) is the marginal financial benefit of retaining worker saving as a
source of investment surplus.  The sign of \(\mu(t)\) determines whether reliance is
constant or declining, while its magnitude determines the reliance reduction rate
through the marginal reliance reduction cost.  Optimal reliance is therefore governed
by the discounted future net marginal cost of worker finance, rather than by the
current-flow comparison alone.

Let
\(g(t)=C_u(I(t),u(t),z(t);x(t))-\sigma^L Y^L(x(t))\).
Then the shadow equation in Lemma~\ref{lem:interior-optimality-conditions} can be
written as
\(\dot{\mu}(t)=r\mu(t)-g(t)\).

The forward-looking nature of the financing decision produces a sharper
implication: optimal reliance can begin to decline before worker finance becomes
unattractive in current-flow terms.

\begin{corollary}
\label{cor:anticipatory-reliance-reduction}
Under the conditions of Lemma~\ref{lem:interior-optimality-conditions}, suppose \(t_0\) is a zero of \(\mu\) on an interior continuation arc, so that
\(\mu(t_0)=0\).  Then \(t_0\) is a regular onset of declining reliance,
\(\dot\mu(t_0)>0\), if and only if
\[
	C_u(I(t_0),u(t_0),z(t_0);x(t_0))
	<
	\sigma^L Y^L(x(t_0)).
\]
Thus a regular onset of declining reliance occurs while worker saving still
provides a positive current net financing benefit.
\end{corollary}

\begin{proof}
At \(\mu(t_0)=0\), the shadow equation gives
\(\dot\mu(t_0)=-g(t_0)\).
Hence \(\dot\mu(t_0)>0\) if and only if \(g(t_0)<0\), which is equivalent to the
stated inequality.
\end{proof}

Optimal reliance can therefore begin to decline in anticipation of future
labour-income erosion, before worker saving has lost its current financial value.

\begin{corollary}
\label{cor:at-most-one-regular-onset}
Under the conditions of Lemma~\ref{lem:interior-optimality-conditions}, suppose \(g(t)\) changes sign once along the relevant path, from negative to
positive.  Then declining reliance has at most one regular onset.
\end{corollary}

\begin{proof}
At any zero of \(\mu\), the shadow equation gives \(\dot\mu=-g\).  Hence a regular
upward crossing requires \(g<0\), while a downward crossing requires \(g>0\).  Two
regular upward crossings would require a downward crossing between them, since
\(\mu\) would have to return to a non-positive value before becoming positive again.
With a single sign change of \(g\) from negative to positive, any regular upward
crossing must occur before the sign change and any downward crossing after it.
Therefore there can be at most one regular onset of declining reliance.
\end{proof}

This is the sense in which reliance can begin to decline before the current flow
comparison turns against worker finance.  The result concerns regular onsets only;
it does not impose a stronger single-switch structure on the entire reliance path.

We thus link the control problem directly to the labour-income analysis.  Along
full-automation paths satisfying the collapse condition in
Proposition~\ref{prop:labour-income-collapse-persistence},
\(Y^L(x(t))\to0\).
Then the financing benefit of worker saving disappears:
\(\sigma^L Y^L(x(t))\to0\).
If the marginal financing cost of reliance on worker finance does not vanish at the
same rate, the future net marginal cost of reliance becomes positive.  The optimal
path then moves toward lower \(u\).  Labour-income collapse therefore reduces
optimal reliance on worker finance by weakening the benefit side of the reliance
trade-off, and it can do so in advance through the forward-looking shadow value
\(\mu\).

Proposition~\ref{prop:artificial-labour-accumulation} shows how accumulation of
artificial labour can affect the two sides of the reliance trade-off differently.
At states with \(B_\Lambda\geq0\), a larger artificial-labour stock moves the
economy toward breakaway.  Holding the external-finance wedge fixed, if
\(\rho>1-\alpha\) the same increase lowers \(Y^L\) and therefore raises \(g\); if
\(\rho<1-\alpha\), it raises \(Y^L\) and lowers \(g\).  These are ceteris-paribus
effects through the labour-income channel; the effect on the shadow value \(\mu\)
also depends on the induced future path and on any state dependence of \(C_u\).

\begin{proposition}
\label{prop:finite-time-zero-reliance}
Suppose the optimal path satisfies the conditions of Lemma~\ref{lem:interior-optimality-conditions} while \(u(t)>0\).  Suppose that, along such an arc, there exist \(T\geq0\) and
\(\varepsilon>0\) such that
\[
	C_u(I(t),u(t),z(t);x(t))
	-
	\sigma^L Y^L(x(t))
	\geq
	\varepsilon
	\qquad
	\text{for all }t\geq T.
\]
Suppose also that there exists a continuous, strictly increasing function
\(\kappa\), with \(\kappa(0)=0\), such that
\[
	C_z(I(t),u(t),z;x(t))
	\leq
	\kappa(z)
\]
for all \(t\geq T\) on the arc and all relevant \(z\geq0\).  Suppose, in addition,
that the terminal condition of Proposition~\ref{prop:forward-looking-reliance}
holds on the arc.  Then, if reliance has not already reached zero by \(T\), it
reaches zero in finite time.
\end{proposition}

\begin{proof}[Proof sketch]
\renewcommand{\qedsymbol}{}
Suppose, to the contrary, that reliance remains positive forever after \(T\). The uniform net-cost condition, together with the forward representation, forces a positive shadow value for reducing reliance. The KKT condition and the bound on the marginal reliance-reduction cost then force reliance to keep falling at a non-vanishing rate. Such a path cannot keep reliance positive forever. A full proof is given in Appendix \ref{app:proof-finite-time-zero-reliance}.
\end{proof}

The onset result and the finite-time result concern distinct margins.
Corollary \ref{cor:anticipatory-reliance-reduction} shows that reliance may begin
to decline before the current flow comparison turns against worker finance.  The
stronger conditions in Proposition \ref{prop:finite-time-zero-reliance} ensure that,
once the net marginal cost of reliance is uniformly positive, the force toward lower
reliance does not vanish before reliance reaches zero.

Without these additional assumptions, reliance may decline gradually and need not
reach zero in finite time.  This distinction has a natural financing interpretation.
If the marginal cost of worker finance vanishes as reliance approaches zero, for
example in the separable benchmark
\(C^R(u)=\gamma u^2/2\),
then \(C_u\to0\) as \(u\to0\), and the force toward lower reliance may also
vanish.  If instead each remaining unit of worker finance carries a positive
marginal financing cost, for example
\(C^R(u)=cu+\gamma u^2/2\) with \(c>0\),
then the cost of maintaining reliance remains positive as long as \(u>0\),
providing the non-vanishing marginal-cost force required by the finite-time result
under the uniform conditions above.

The economic implication is direct.  As automation erodes labour income, the
financing benefit generated by worker saving falls relative to the external-finance
wedge.  Because the comparison is forward-looking, reliance may begin to decline
before worker finance has ceased to be worthwhile in the current flow.  The
breakaway surplus \(B(x)\) answers the separate current-flow self-financing
question; Definition~\ref{def:optimal-reliance-problem} asks whether worker finance
remains optimally valuable once that financing restriction is removed.

\section{Conclusion}
\label{sec:conclusion}

Human labour can remain economically important through several distinct channels
as automation advances. It can contribute directly to production, generate saving
that finances accumulation, and remain valuable within the human-filled activities
that survive. This paper studies the economic conditions governing these margins
and, in particular, the point at which continued automation becomes financially
independent of saving generated by labour income. We call this transition
automation breakaway.

The central object is the breakaway surplus,
\(B(x)=\Sigma^E(x)-\mathcal{T}^{\mathrm{full}}(x)\),
which compares saving generated by capital and artificial labour with the resources
required to maintain the automation system and continue advancing the frontier.
Viability asks whether the dynamics preserve this financing capacity. The remaining
results then characterise how this accumulation margin relates to labour income,
surviving human work, and optimal reliance on worker finance.

Four economic separations emerge. Technological independence is compatible with
continued use of labour-income saving. Accumulation independence is compatible
with either rising or falling labour income. Aggregate labour-income collapse is
compatible with positive, and potentially increasing, average labour income in the
human-filled jobs that remain. Worker finance can retain positive optimal value
when its use is optional. The formal results identify the economic conditions
supporting each of these outcomes.

Two mechanisms are especially important for breakaway. Productivity changes the
amount of income generated by the economy, while income composition changes the
share accruing to the accumulable bloc. Viable breakaway is governed by a marginal
investment condition: productivity contributes through its effect on output, while
income composition provides an independent channel through which frontier
advancement can generate the saving required for its own continuation.
Artificial-labour accumulation reinforces this result by expanding the non-labour
saving base while its effect on labour income depends on substitutability between
human and artificial labour. The labour-market results add a further margin:
aggregate human work can contract while the economic value of the human-filled jobs
that remain stays high.

The framework isolates these mechanisms while providing natural interfaces with
richer models. Endogenous saving and asset prices introduce the accumulation
feedbacks studied by \citet{caballero2026speculative}. Under the maintained
income-source saving convention, saving rates by income source determine the
breakaway surplus, while the allocation of ownership claims determines who receives
the corresponding income. \citet{lagarda2026owns} provides a natural ownership
interface when saving behaviour depends on who receives automation income.
Financial frictions provide another interface, as in \citet{bond2023income}
(revised 2026). Preferences and distributional objectives provide the additional
structure required for welfare analysis.

The broader implication is that the transition toward a post-labour economy has
several economically distinct dimensions. Automation changes the technological
role of human labour, the sources that finance accumulation, the evolution of
labour income, the value of surviving human work, and optimal reliance on worker
finance. Automation breakaway identifies the financing dimension of this
transition. An automated production system and a self-financing automation system
describe different economic states, governed by different conditions.

\clearpage
\appendix
\section{Proofs}
\label{app:proofs}

The main text retains proofs that reveal the economic or mathematical mechanism of a result. Proofs consisting primarily of algebra, standard convergence arguments, or technical verification are collected here; where the mechanism is informative but the verification is longer, the main text gives a proof sketch.

\subsection{Labour-Income Collapse and Persistence}
\label{app:proof-labour-income-collapse}

\begin{proof}[Proof of Proposition~\ref{prop:labour-income-collapse-persistence}]
By Equation~\eqref{eq:labour-income-asymptotic} and the definition of \(H(t)\),
\[
\frac{Y^L(x(t))}
{H(t)h(K(t),L(t),\Lambda(t))}
\to1.
\]
Hence, for all sufficiently large \(t\),
\[
\frac{1}{2}H(t)h(K(t),L(t),\Lambda(t))
\leq
Y^L(x(t))
\leq
\frac{3}{2}H(t)h(K(t),L(t),\Lambda(t)).
\]
Also,
\begin{equation}
\label{eq:human-labour-margin-growth}
H(t)h(K(t),L(t),\Lambda(t))
=
H(0)h(K(0),L(0),\Lambda(0))
\exp\left\{
-\int_0^t
\left[
\rho_\beta(s)
-\theta\rho_a(s)
-\rho_h(s)
\right]ds
\right\}.
\end{equation}
The two equivalences follow from these bounds and the exponential representation.
\end{proof}

\subsection{Task Automation and Human-Filled Jobs}
\label{app:proof-task-automation-jobs}

\begin{proof}[Proof of Proposition~\ref{prop:task-automation-jobs}]
	For each \(\tau\), \(1-\beta^\tau\to0\) as \(\beta\to1\), and
	\(0\leq1-\beta^\tau\leq1\).  Dominated convergence therefore gives
	\(J^H(\beta)\to0\).
	If \(\bar\tau<\infty\), then
	\((1-\beta^\tau)/(1-\beta)\to\tau\)
	for each \(\tau\).  Moreover, for \(\beta\in(0,1)\),
	\[
		0
		\leq
		\frac{1-\beta^\tau}{1-\beta}
		=
		1+\beta+\cdots+\beta^{\tau-1}
		\leq
		\tau .
	\]
	Dominated convergence gives, as $\b\to 1$
	\[
		\frac{J^H(\beta)}{1-\beta}
		=
		\sum_{\tau\geq1}p_\tau
		\frac{1-\beta^\tau}{1-\beta}
		\to
		\sum_{\tau\geq1}\tau p_\tau
		=
		\bar\tau .
	\]
\end{proof}

\subsection{Artificial-Labour Accumulation}
\label{app:proof-artificial-labour-accumulation}

\begin{proof}[Proof of Proposition~\ref{prop:artificial-labour-accumulation}]
Since
\[
\frac{\partial\log Y}{\partial\log\Lambda}
=(1-\alpha)q,
\qquad
\frac{\partial\log q}{\partial\log\Lambda}
=\rho(1-q),
\]
differentiating
\(\Sigma^E=\left[\sigma^K\alpha+\sigma^\Lambda(1-\alpha)q\right]Y\)
gives Equation~\eqref{eq:SigmaE-lambda}.  Similarly, from
\[
Y^L
=
(1-\alpha)Y
\frac{(1-\beta)a^\theta L^\rho}{M^\rho}
\]
and
\(
\partial\log M/\partial\log\Lambda=q
\),
\[
\frac{\partial\log Y^L}{\partial\log\Lambda}
=
(1-\alpha)q-\rho q
=
(1-\alpha-\rho)q,
\]
which gives Equation~\eqref{eq:YL-lambda}.  Equation~\eqref{eq:B-lambda}
follows by differentiating the linear-benchmark breakaway surplus.
\end{proof}

\subsection{Finite-Time Zero Reliance}
\label{app:proof-finite-time-zero-reliance}

\begin{proof}[Proof of Proposition~\ref{prop:finite-time-zero-reliance}]
Suppose by contradiction that reliance never reaches zero after \(T\).  Then the
path remains on an interior arc with \(u(t)>0\) for all \(t\geq T\), so
Proposition~\ref{prop:forward-looking-reliance} applies and
\(\mu(t)\geq\varepsilon/r\) for all \(t\geq T\).
Hence \(\mu(t)>0\), so reliance is declining for all \(t\geq T\), with
\[C_z(I(t),u(t),z(t);x(t))=\mu(t).\]
The bound on \(C_z\) therefore implies
\(\kappa(z(t))\geq\varepsilon/r\),
and hence
\[
	z(t)\geq \kappa^{-1}\left(\frac{\varepsilon}{r}\right)
	=: \underline z>0
\]
as long as \(u(t)>0\).  But then
\(\dot u(t)=-z(t)\leq-\underline z\)
for all \(t\geq T\) on the interior arc, contradicting \(u(t)>0\) forever.
\end{proof}
\bibliographystyle{plainnat}
\bibliography{references}

\end{document}